\documentclass[12pt,a4paper]{article}
\usepackage[utf8]{inputenc}
\usepackage{amssymb,amsfonts,amsmath}
\usepackage{subfigure}
\usepackage{graphicx}
\usepackage{hyperref}
\usepackage{physics}
\usepackage{dsfont}
\usepackage{bbold}
\usepackage{amsthm}
\usepackage{color}

\usepackage[
top    = 3.2cm,
bottom = 3.0cm,
left   = 2.8cm,
right  = 2.8cm]{geometry}

\newcommand\be{\begin{equation}}
\newcommand\ee{\end{equation}}
\newcommand\ba{\begin{eqnarray}}
\newcommand\ea{\end{eqnarray}}
\newcommand{\fract}[2]{{\textstyle\frac{#1}{#2}}}
\newcommand{\bb}[1]{\mathbb{#1}}

\newtheorem{prop}{Proposition}
       \newtheorem{thm}{Theorem}[section]
       \newtheorem{lemma}[thm]{Lemma}
       
       \theoremstyle{remark}
        \newtheorem*{remark}{{\bf Remark}}
         \newtheorem*{example}{{\bf Example}}
    
\begin{document}

\begin{titlepage}

\begin{center}

{\Large \bf Drinfel'd-Sokolov construction and exact solutions of  vector modified KdV hierarchy}

\vskip 1.5cm

{{\bf Panagiota Adamopoulou$^{\dagger}$ and Georgios Papamikos$^{\star}$ }} 

\vskip 0.8cm

{\footnotesize
$^{\dagger}$School of Mathematical \& Computer Sciences, Heriot-Watt University,\\
Edinburgh EH14 4AS, United Kingdom}
\\
{\footnotesize
$^{\star}$School of Mathematics, University of Leeds, Leeds LS2 9JT, United Kingdom}

\vskip 0.5cm

{\footnotesize {\tt E-mail: p.adamopoulou@hw.ac.uk, g.papamikos@leeds.ac.uk }}\\

\end{center}

\vskip 2.0cm

\centerline{\bf Abstract}

We construct the hierarchy of a multi-component generalisation of modified KdV equation and find exact solutions to its associated members. The construction of the hierarchy and its conservation laws is based on the Drinfel'd-Sokolov scheme, however, in our case the Lax operator contains a constant non-regular element of the underlying Lie algebra. We also derive the associated recursion operator of the hierarchy using the symmetry structure of the Lax operators. Finally, using the rational dressing method, we obtain the one soliton solution, and we find the one breather solution of general rank in terms of determinants.

\vfill

\end{titlepage}

\section{Introduction}

In this paper we study the vector modified Korteweg-de Vries equation (vmKdV)
\be \label{vmkdv}
{\bf u}_{t} + {\bf u}_{xxx} + \fract{3}{2} \| {\bf u}\| ^2 {\bf u}_x = 0 \,, \quad \mbox{with} \quad {\bf u} = {\bf u}(x,t) \in \mathbb{R}^N \,.
\ee
In what follows, we denote vectors by boldface, and the upper index $T$ denotes transposition, so that the standard Euclidean norm  is $\| {\bf u} \| =\sqrt{{\bf u}^T{\bf u}}$. When $N=1$, equation \eqref{vmkdv} reduces to the well-known modified KdV equation (mKdV)
$$u_t + u_{xxx} + \fract{3}{2} u^2u_x=0 \,,$$
while for $N=2$ one obtains the complex mKdV (also known as Hirota equation \cite{Hirota eq}) 
$$v_t + v_{xxx} + \fract{3}{2} |v|^2 v_x=0$$ 
for the complex field $v=u_1 + i u_2$.

Equation \eqref{vmkdv} appeared in \cite{JP vmkdv 1, JP vmkdv} in the study of the evolution of a curve in a $(N+1)$-dimensional Riemannian manifold, while several generalisations of the mKdV equation have been introduced by various authors such as, for example, Iwao and Hirota \cite{Iwao-Hirota}, Fordy and Athorne \cite{Fo-Ath} in association with symmetric spaces, Sokolov and Wolf \cite{Sokolov Wolf} using the symmetry approach, Svinolupov and Sokolov \cite{Svino, Svino-Sok} in relation to Jordan algebras,  and non-associative algebras in general  \cite{Sokolov vmkdv}. Moreover, the study of soliton solutions and their interactions for several multi-component generalisations of the scalar mKdV equation have been studied using the inverse scattering transform, or the Hirota, dressing, or other methods, see for example  \cite{Anco sol, Iwao-Hirota, Tsuchida-Wadati, Tsuchida, fenchenko khruslov, pelinovsky}. For the case of the vmKdV equation \eqref{vmkdv}, its integrability properties, such as the Lax representation, recursion operator and Hamiltonian structure, where presented by Sanders and Wang in \cite{JP vmkdv}, while the recursion operator and Hamiltonian structure where further investigated by Anco in \cite{Anco vmkdv}. Moreover, in \cite{miky titi}, efficient numerical integration schemes of \eqref{vmkdv} where considered. We note here that the vmKdV equation which is the third member of the vector nonlinear Schr{\"o}dinger (vNLS) hierarchy (see for example \cite{vNLS}) is of the form
$$
{\bf u}_t + {\bf u}_{xxx} +\fract{3}{4}{\bf u}^T {\bf u}_x {\bf u} + \fract{3}{4} \|{\bf u}\|^2 {\bf u}_x= 0 \,,
$$
thus, different from the vmKdV equation that we consider in this paper. A relation between \eqref{vmkdv} and the nonlinear Schr{\"o}dinger equation (NLS) is discussed in Section \ref{sec: so4}.

In this work we construct the hierarchy associated to equation \eqref{vmkdv}, and derive the one soliton and one breather solution of general rank for the whole hierarchy. A method for constructing integrable hierarchies based on a given Lie algebra $\mathfrak{g}$ was introduced by Drinfel'd and Sokolov in \cite{DS1, DS}, and later explored further in e.g. \cite{Wilson, deGroot, Feher, Mikh DS}, and relies on the idea of dressing or formal (in the spectral parameter) Darboux transformations. A basic ingredient in the construction presented in these works is a Lax operator containing a constant regular element of $\mathfrak{g}$. Here, the construction of the hierarchy and its corresponding conservation laws are derived based on the ideas of the seminal works \cite{DS1, DS}, considering, however, a Lax operator with a constant non-regular element of the underlying Lie algebra. An additional property of the Lax operator is the invariance under the Cartan involution introduced in Section \ref{sec: Lax}. In particular, in Section \ref{sec: Lax} we present the Lax structure for the vmKdV equation \eqref{vmkdv} and its associated symmetries that will be extensively used in the following sections. Then, Section \ref{sec: vmKdV hier} is dedicated to the construction of the vmKdV hierarchy and its conservation laws, while the commutativity of the members of the hierarchy is also proved.  The recursion operator for the hierarchy is also derived using an alternative method than those used in \cite{JP vmkdv, Anco vmkdv}. As mentioned previously, a relation to the NLS equation is discussed in Section \ref{sec: so4}. Finally, in Section \ref{sec: dressing}, we employ the method of rational dressing \cite{Zakharov Shabat II, Zakharov Mikhailov}, which is based on the concept of Darboux transformations, in order to construct  solutions for the hierarchy. We derive the one soliton solution, as well as the one breather solution of general rank  in terms of determinants.

\section{Lax structure} \label{sec: Lax}

We consider the following differential operator
\be\label{L}
\mathcal{L}(\lambda) = D_x - \mathcal{U}(\lambda) \quad \mbox{with} \quad \mathcal{U}(\lambda) = \lambda J + U \,,
\ee
which constitutes the spatial part of the Lax pair associated to the vmKdV equation \eqref{vmkdv}, see \cite{JP vmkdv}. Here $J$ and $U$ are elements of the Lie algebra $\mathfrak{g} :=  \mathfrak{so}_{N+2}(\mathbb{R})$ and are of the form
\be \label{J U}
J = \begin{pmatrix}
0 & 1 & {\bf 0}^T \\
-1 & 0 & {\bf 0}^T \\
{\bf 0} & \bf 0 & \bb{0}  \\
\end{pmatrix},
\quad
U = \begin{pmatrix}
0 & 0 & {\bf 0}^T \\
0 & 0 & {\bf u}^T \\
{\bf 0} & - \bf u & \bb{0}  \\
\end{pmatrix} \,,
\ee
where ${\bb 0}$ stands for the $N \times N$ zero matrix. When there is no ambiguity, we will use $0$ to denote a square zero matrix of any dimension. Also, $\lambda \in \overline{\mathbb{C}}=\mathbb{C}\cup\lbrace \infty \rbrace$ is a spectral parameter of  the linear spectral problem $\mathcal{L}(\lambda) \Psi = 0$, $\mathcal{U}(\lambda)$ is an element of the loop algebra $\mathfrak{g}[\lambda]$, and $\mathcal{L}(\lambda)$ is in the ring of differential operators $\mathfrak{g}[\lambda][D_x]$.

The Lax operator \eqref{L} is invariant under the action of a group of automorphisms generated by the following transformations
\ba
\mathfrak{t}: & \mkern-35mu \mathcal{L}(\lambda) \rightarrow - \mathcal{L}(\lambda)^{\dagger} \,, \label{t-aut}\\
\mathfrak{q}: &  \mathcal{L}(\lambda) \rightarrow  Q \mathcal{L}(-\lambda) Q^{-1} \,, \label{q-aut}\\
\mathfrak{r}: & \mkern-48mu   \mathcal{L}(\lambda) \rightarrow \mathcal{L}(\lambda^*)^{*} \label{r-aut} \,.
\ea
Here, $\mathcal{L}(\lambda)^{\dagger} := -D_x - \lambda J^T - U^T$ denotes the formal adjoint of $\mathcal{L}(\lambda)$,  the $*$ denotes complex conjugation, and $Q = \mbox{diag}(-1,1, \ldots,1)$. The reality requirement for the entries of $U$ is equivalent to the invariance of $\mathcal{L}$ under the action of $\mathfrak{r}$. Invariance of $\mathcal{L}$ under $\mathfrak{t}$ implies that the matrices $J, \, U$ are skew symmetric. Since the transformations $\mathfrak{t}, \mathfrak{q}, \mathfrak{r}$ are involutions and commute with each other they generate the group $\mathbb{Z}_2 \times \mathbb{Z}_2 \times \mathbb{Z}_2$. Such symmetry groups are known as reduction groups \cite{Mikh RG2, Mikh RG3} and have been extensively used in the theory of integrable systems, see for example \cite{Mikh RG1, Belavin Drinfeld}.

The automorphism $a\rightarrow QaQ^{-1}$ of the Lie algebra $\mathfrak{g}$ is an involution and hence its eigenvalues are $\pm 1$. We define the corresponding eigenspaces
\begin{equation}
\mathfrak{g}^{(j)}=\lbrace a\in \mathfrak{g}~|~QaQ^{-1}=(-1)^ja\rbrace, \quad j\in\mathbb{Z}_2,
\label{eq:eigenspac}
\end{equation}
which satisfy the following commutation relations
\begin{equation} \label{Cartan dec}
\left[\mathfrak{g}^{(i)},\mathfrak{g}^{(j)}\right]\subseteq \mathfrak{g}^{(i+j)},
\end{equation}
and thus define the $\mathbb{Z}_2$-grading $\mathfrak{g}=\mathfrak{g}^{(0)} \oplus \mathfrak{g}^{(1)}$ known as Cartan decomposition of a Lie algebra \cite{Helgason}. Such reductions to symmetric spaces have been used in \cite{Fo-Ku, Fo-Ath} to construct multi-component integrable equations of KdV, mKdV, and NLS type.  They have also appeared in \cite{vSG dress, vSG DT} in relation to a vectorial generalisation of the sine-Gordon equation. Elements that belong to either $\mathfrak{g}^{(0)}$ or $\mathfrak{g}^{(1)}$ are called homogeneous, and in particular elements of $\mathfrak{g}^{(0)}$ are called even while those of $\mathfrak{g}^{(1)}$ are called odd, and they are of the form
\begin{equation}
\left(
\begin{array}{ccc}
0 & 0 & {\bf 0}^T\\
0 & 0 & *_1^T\\
{\bf 0} & -*_1 & *
\end{array}
\right),
\quad 
\left(
\begin{array}{ccc}
0 & *_1 & *^T_2\\
-*_1 & 0 & {\bf 0}^T\\
-*_2 & {\bf 0} &  \bb{0} 
\end{array}
\right),
\label{eq:even-odd}
\end{equation}
respectively. In particular, in \eqref{J U} $U$ is even while $J$ is odd. We further define the map $\mbox{ad}_J(a) := \left[ J , a \right], \forall a \in \mathfrak{g}$, and the spaces $\mathfrak{K} = \text{Ker ad}_J$ and $\mathfrak{I} =  \text{Im ad}_J$, such that 
\be \label{alg dec}
\mathfrak{g} = \mathfrak{K} \oplus  \mathfrak{I} \,.
\ee
Using the $\mathbb{Z}_2$ gradation given above, we further define the  spaces $\mathfrak{K}^{(i)} = \mathfrak{K} \cap \mathfrak{g}^{(i)}$ and $\mathfrak{I}^{(i)} = \mathfrak{I} \cap \mathfrak{g}^{(i)}$, with $i \in \mathbb{Z}_2$, which obey the following commutation relations
\be \label{Ker dec}
[ \mathfrak{K}^{(1)}, \mathfrak{K}^{(1)} ]=0 \,, ~ [ \mathfrak{K}^{(0)}, \mathfrak{K}^{(0)} ] \subseteq  \mathfrak{K}^{(0)}, ~ [ \mathfrak{K}^{(0)}, \mathfrak{K}^{(1)} ]=0 \,,
\ee
and thus provide the decomposition $\mathfrak{K} = \mathfrak{K}^{(0)} \oplus \mathfrak{K}^{(1)}$. We remark here that the construction in which $\mathfrak{K}$ is an abelian subalgebra of $\mathfrak{g}$ was presented in \cite{DS}, as well as used in later works e.g. \cite{Fo-Ku, Wilson, deGroot}. In the current situation, although $\mathfrak{K}$ is not abelian, the decomposition of $\mathfrak{K}$ to odd and even parts is crucial for the construction of the vmKdV hierarchy and its conservation laws. In the case where $N=2$, $\mathfrak{K}^{(0)}$ is abelian and thus from \eqref{Ker dec} it follows that $\mathfrak{K}$ is also abelian. We discuss this situation further in Section \ref{sec: so4}. We will use the above definitions together with the properties of the $\mathcal{L}(\lambda)$ operator \eqref{L} in the following section, where we construct the vmKdV hierarchy  following the ideas of Drinfel'd and Sokolov.

\section{The vector mKdV hierarchy} \label{sec: vmKdV hier}

The vmKdV hierarchy is defined as the set of all compatibility conditions $[\mathcal{L}(\lambda), \mathcal{A}_k(\lambda)] = 0$, $k \in \mathbb{Z}_{+}$,  of the systems of spectral problems $\mathcal{L}(\lambda) \Psi = 0$, $\mathcal{A}_k(\lambda) \Psi = 0$, where $\mathcal{L}(\lambda)$ is given in \eqref{L} and $\mathcal{A}_k(\lambda)$ are of the form 
\be\label{A}
\mathcal{A}_k(\lambda) = D_{t_k} - \mathcal{V}_k(\lambda). 
\ee
Each $\mathcal{V}_k(\lambda)$ is an appropriately chosen matrix-valued polynomial in $\lambda$ of degree $k$. The case where $k=-1$ was considered in \cite{vSG dress, vSG DT} in relation to the vector sine-Gordon equation.

For given $k$, the compatibility condition $[\mathcal{L}, \mathcal{A}_k] = 0$  is equivalent to the Lax equation
\be\label{Lax eq}
D_{t_k} \mathcal{L} = [  \mathcal{V}_k, \mathcal{L}  ] \,.
\ee
The left-hand side of equation \eqref{Lax eq} is equal to
\be \label{Lt}
D_{t_k} \mathcal{L} = - D_{t_k} U= - \begin{pmatrix}
0 & 0 & {\bf 0}^T \\
0 & 0 &  {\bf u}_{t_k}^T \\
{\bf 0} &  - {\bf u}_{t_k} & \bb{0}  \\
\end{pmatrix} \,, 
\ee
therefore, if the commutator $ [  \mathcal{V}_k, \mathcal{L}  ] $ is of the same form as \eqref{Lt} then we obtain consistent evolutionary  equations, integrable in the Lax sense. Hence, the problem of determining the vmKdV hierarchy reduces to finding all those $\mathcal{V}_k$ such that $[  \mathcal{V}_k, \mathcal{L}  ]$ is $\lambda$-independent and skew-symmetric  matrix of the form \eqref{Lt}. To this end, we define the set of formal series 
$$
\mathfrak{g}(( \lambda^{-1}))^{\langle \mathfrak{q} \rangle} :=\bigg \lbrace \sum_{i = -\infty}^{l} a_i \lambda^i \,, a_i \in \mathfrak{g} \,, l \in \mathbb{Z}_+ \,, \mathfrak{q}(a_i)= (-1)^{i\bmod{2}}a_i \bigg  \rbrace
$$ 
and the projections of $a(\lambda) \in \mathfrak{g}(( \lambda^{-1}))^{\langle\mathfrak{q}\rangle}$ to its polynomial in $\lambda$ part $a(\lambda)_{+}$, and the part consisting of strictly negative powers in $\lambda$, $a(\lambda)_{-}=a(\lambda)-a(\lambda)_+$. The following Lemma gives sufficient conditions for the characterisation of the $\mathcal{V}_k(\lambda)$ as described above.

\begin{lemma}\label{lem split}
Let $a(\lambda) = \sum_{i = - \infty}^{k} a_i \lambda^i \in \mathfrak{g}(( \lambda^{-1}))^{\langle \mathfrak{q} \rangle}$, with $a_k$ non-zero, and such that $[a(\lambda), \mathcal{L}(\lambda)] = 0$. Then, $[a(\lambda)_{+}, \mathcal{L}(\lambda)]$ is $\lambda$-independent and of the form \eqref{Lt}. Moreover, if $k$ is an odd positive integer then $a_k=J$.  
\end{lemma}

\begin{proof}
The condition $[a(\lambda), \mathcal{L}(\lambda)] = 0$ implies that $[ a(\lambda)_{+}, \mathcal{L}(\lambda) ] = - [a(\lambda)_{-}, \mathcal{L}(\lambda)]$. The left-hand side of this relation is polynomial in $\lambda$, while the right-hand side contains non-positive powers of $\lambda$. It follows that
\be \label{res}
[a(\lambda)_{+}, \mathcal{L}(\lambda)] = - [J,a_{-1}]  \in \mathfrak{I}^{(0)}
\ee
and thus is of the same form as \eqref{Lt}. Then, expanding $[a(\lambda),\mathcal{L}(\lambda)]=0$ in $\lambda$ we obtain the following equations for the coefficients $a_i$ of $a(\lambda)$
\be \label{a i}
 \left[ J, a_{k} \right] = 0 \,, \quad 
D_{x} a_i = \left[U, a_i \right] +  \left[ J, a_{i-1} \right] \,, \quad i \leq k \,. 
\ee
If $k$ is odd, it follows that $a_k \in \mathfrak{K}^{(1)} = \mbox{span}_{\mathbb{R}}(J)$, thus $a_k = c \, J$, with $c$ a scalar function. From the second equation in \eqref{a i}, for $i=k$ we obtain that $D_x c=0$. Hence, since $c$ is $x$-independent it can be set to one without affecting the Lax equation. 
\end{proof}

From the discussion so far it follows that those $\mathcal{V}_k(\lambda)$ which give  consistent  Lax equations \eqref{Lax eq} are of the form $a(\lambda)_{+}$, with $a(\lambda)$ as in Lemma \eqref{lem split}. Moreover, above we show that if $k$ is odd then $a_k$ is determined uniquely. If $k$ is even then $a_k \in \mathfrak{K}^{(0)}$ and then equations \eqref{a i} will involve terms which are in the image of $D_x^{-1}$, hence the Lax equations \eqref{Lax eq} will be non-local. Such types of  evolutionary non-local equations were studied in, for example, \cite{boomerons} and have been shown to admit very interesting solutions (boomerons). Also, in \cite{Fo-Ku} the existence of such non-local equations within the NLS hierarchy was remarked. In this work we will primarily consider the case of odd $k$ where $a_k = J$.

From Lemma \ref{lem split} it follows that in order to determine $a(\lambda)$ we have to study equation $[a(\lambda), \mathcal{L}(\lambda)] = 0$, which is equivalent to the Lax equation
\be\label{Lax a}
D_x a(\lambda) = \left[ \mathcal{U}(\lambda), a(\lambda) \right] \,.
\ee
Equation \eqref{Lax a} implies that
\be \label{sol Lax 1}
a(\lambda) = P(\lambda) C(\lambda) P(\lambda)^{-1} \,, \quad \mbox{with} \quad D_x P(\lambda) = \mathcal{U}(\lambda) P(\lambda) 
\ee
and $C(\lambda)$ a constant matrix with respect to $x$. Since we are interested in autonomous equations we also assume that $C(\lambda)$ is independent of $t_k$. Moreover, the fundamental solution $P(\lambda)$ of the differential equation \eqref{sol Lax 1} is an element of the corresponding loop group
$$
L G^{\langle \mathfrak{q} \rangle}  := \lbrace  A(\lambda) \in SO_{N+2}(\mathbb{R})\,, \: \lambda \in \overline{\mathbb{C}} \,, \mathfrak{q}\left(A(\lambda) \right)  = A(\lambda)       \rbrace \,.
$$
Assuming that $U$ is bounded as $\lambda \rightarrow \infty$,  $P(\lambda)$ with initial condition $P(x=0, \lambda) = \mathbb{1}$ behaves as $P(\lambda) \sim \mbox{e}^{\lambda x J}$. This asymptotic behaviour suggests the transformation 
\be \label{P M exp}
P(\lambda) = \mathcal{M}(\lambda) \mbox{e}^{\lambda x J} \,,
\ee
where $\mathcal{M}(\lambda) = P(\lambda) \mbox{e}^{-\lambda x J} \in LG^{\langle \mathfrak{q} \rangle} $  is of the form
\be \label{M lambda}
\mathcal{M}(\lambda) = \mathbb{1} + \lambda^{-1} \mathcal{M}_1 + \cdots \,.
\ee  
Substituting \eqref{P M exp} in the differential equation \eqref{sol Lax 1} we obtain the following equation for $\mathcal{M}(\lambda)$ 
\be \label{D-S 1}
D_{x} \mathcal{M} + \lambda \left[ \mathcal{M}, J \right] = U \mathcal{M} \,.
\ee
Hence,  the Lax matrices $\mathcal{V}_k(\lambda)$ are of the form $\mathcal{V}_k(\lambda) = \left( \mathcal{M}(\lambda) e^{\lambda x J} C(\lambda) e^{-\lambda x J} \mathcal{M}(\lambda)^{-1} \right)_{+}$.

\begin{prop} \label{prop M sol}
Equation \eqref{D-S 1} admits a formal solution $\mathcal{M}(\lambda) \in LG^{\langle \mathfrak{q} \rangle}$ of the form \eqref{M lambda}, where the coefficients $\mathcal{M}_n$ can be found recursively.
\end{prop}

\begin{proof}
Comparing powers of $\lambda$ in  \eqref{D-S 1} we obtain the following equations for $\mathcal{M}_n$
\be \label{Mi} 
\mbox{ad}_J (\mathcal{M}_1) = -U \,,\quad \mbox{ad}_J (\mathcal{M}_{n+1}) = D_x \mathcal{M}_{n} - U \mathcal{M}_n \,, \quad n=1,2,\ldots \,.
\ee
To solve the above equations we use the direct sum decomposition \eqref{alg dec}. The restriction of $\text{ad}_J$ to $\mathfrak{I}$ is invertible with the inverse given by
$$ 
\mbox{ad}_J^{-1} = -\frac{5}{4} \mbox{ad}_J - \frac{1}{4} \mbox{ad}_J^3 \,.
$$
Moreover, it follows that the projectors to $\mathfrak{I}$ and $\mathfrak{K}$ are 
$$
P_{\mathfrak{I}} = -\frac{5}{4} \mbox{ad}_J^2 - \frac{1}{4} \mbox{ad}_J^4 \quad \text{and} \quad P_{\mathfrak{K}} = \mbox{id}- P_{\mathfrak{I}} \,, 
$$
respectively. In order to determine the $\mathcal{M}_n$ we apply projectors $P_{\mathfrak{I}}$ and $P_{\mathfrak{K}}$ in equations \eqref{Mi} and solve them recursively. In particular,
$$
P_{\mathfrak{I}}(\mathcal{M}_{n+1})=\mbox{ad}_J^{-1}(D_x\mathcal{M}_n-U\mathcal{M}_n).
$$
The projections in  ${\mathfrak{K}}$ will lead to integrations and thus non-local terms in ${\bf u}$ and its derivatives in $x$. Any remaining freedom in the components of $P_{\mathfrak{K}}(\mathcal{M}_n)$ can be fixed by requiring  $\mathcal{M}(\lambda)\in LG^{\langle \mathfrak{q}\rangle}$. 
\end{proof}

Equation \eqref{D-S 1} can be re-written as
\be \label{intertw}
\mathcal{L}(\lambda)\mathcal{M}(\lambda)=\mathcal{M}(\lambda)\mathcal{L}_0(\lambda)
\ee
where $\mathcal{L}_0(\lambda)=D_x-\lambda J$. This means that the matrix $\mathcal{M}(\lambda)$ is a formal Darboux-Dressing matrix \cite{matveev salle, Zakharov Shabat II} of the trivial $\mathcal{L}_0(\lambda)$ operator which corresponds to the potential $\textbf{u}_0=\textbf{0}$. In Section \ref{sec: dressing}, we construct a closed form Darboux matrix $\mathcal{M}(\lambda)$  which maps $\textbf{u}_0$ to the one soliton and one breather solution of the vmKdV hierarchy. Next we prove that the formal Darboux matrix $\mathcal{M}(\lambda)$ admits a  factorisation which simplifies the construction of the vmKdV hierarchy and also provides its conservation laws. Such factorisations also appear in e.g. \cite{DS, Mikh DS}.

\begin{prop} \label{prop M=WH}
The Darboux matrix $\mathcal{M}(\lambda)$ can be factorised as 
\be\label{factor DT} 
\mathcal{M}(\lambda) = W(\lambda) H(\lambda) \,,
\ee
where $W(\lambda)$, $H(\lambda)\in LG^{\left< \mathfrak{q} \right>}$ and are of the form 
\be \label{W H asym}
W(\lambda) = \mathbb{1} + \frac{1}{\lambda} W_1 + \cdots \,, \quad H(\lambda) = \mathbb{1} + \frac{1}{\lambda} H_1 + \cdots \,, 
\ee
satisfying $H(\lambda)J H(\lambda)^{-1} = J$ and 
\be \label{D-S 2}
D_{x} W + W h + \lambda \left[ W, J \right] = U W\,, \quad h = H_x H^{-1} \in \mathfrak{K}((\lambda^{-1}))^{\langle \mathfrak{q} \rangle} \,.
\ee
\end{prop}

\begin{proof}
If the Darboux matrix \eqref{M lambda} is written as the product of $W(\lambda)$ and $H(\lambda)$ which are elements of $LG^{\left< \mathfrak{q} \right>}$ for all $\lambda$, and of the form \eqref{W H asym}, then equation \eqref{D-S 1} takes the form \eqref{D-S 2}, assuming that $H(\lambda)J H(\lambda)^{-1} = J$ for all $\lambda$. It follows from \eqref{W H asym} that 
\be \label{h}
h(\lambda) = H_xH^{-1} = \frac{1}{\lambda}h_1 + \frac{1}{\lambda^2}h_2 + \cdots \,,
\ee
with each $h_i \in \mathfrak{K}^{(i \bmod 2)}$. Comparing powers of $\lambda$ in equation \eqref{D-S 2}, we obtain the following relation for the coefficients of $W(\lambda)$ and $h(\lambda)$ at $\lambda^{-n}$
\be
 \sum_{i=0}^{n}W_i h_{n-i}  + \left[ W_{n+1}, J \right]  - U W_n +  D_x W_n =0 \,, \quad  n=0,1, \ldots \,, \label{D-S n}
 \ee
with $W_0 = \mathbb{1}$ and $h_0 = 0$. Hence, the coefficients  $W_i$ and $h_i$ can be recursively determined from \eqref{D-S n} using the decomposition \eqref{alg dec} and requiring that $W(\lambda)\in LG^{\left< \mathfrak{q} \right>}$. The factorisation \eqref{factor DT} is not unique since we have that the transformation $(W(\lambda),H(\lambda)) \rightarrow (W(\lambda) S(\lambda), S(\lambda)^{-1} H(\lambda))$, with $S(\lambda)= S_0 + \frac{1}{\lambda} S_1 + \cdots$, leaves $\mathcal{M}(\lambda)$ invariant. We can use this gauge freedom to fix the asymptotic behaviour of $W(\lambda)$ and $H(\lambda)$ at $\lambda = \infty$. Finally, any remaining freedom in $W(\lambda)$ is fixed by additionally requiring $[S(\lambda),J]=0$.
\end{proof}

The following proposition determines the Lax matrices $\mathcal{V}_k(\lambda)$ and effectively defines the vmKdV hierarchy.

\begin{prop} \label{prop Vk}
If $\mathcal{V}_k (\lambda)$ is given by
\be \label{Vk}
\mathcal{V}_k (\lambda) =  \left(\lambda^k \mathcal{M}(\lambda) J \mathcal{M}(\lambda)^{-1} \right)_{+} = \left( \lambda^k W(\lambda) J W(\lambda)^{-1} \right)_{+} \,, \quad k=1,3, \ldots \,,
\ee
where $W(\lambda)$ is as in Proposition \ref{prop M=WH}, then the Lax equations \eqref{Lax eq} are local evolutionary PDEs which admit the Lie symmetry
\begin{equation}
(x,t_k,u) \mapsto (\widetilde{x},\widetilde{t}_k,\widetilde{u})=(e^{\epsilon}x,e^{k\epsilon}t_k,e^{-\epsilon}u), \quad k=1,3,...,
\label{scale-k}
\end{equation}
respectively.
\end{prop}

\begin{proof}
The members of the vmKdV hierarchy are given by the Lax equations \eqref{Lax eq}, which can also be written as $-U_{t_k}= [  \mathcal{V}_k, \mathcal{L}  ]$. From Lemma \ref{lem split} we have that $\mathcal{V}_k(\lambda) = a(\lambda)_{+}=(P(\lambda)C(\lambda)P(\lambda)^{-1})_{+}$ with $a(\lambda) \sim \lambda^k J$  as $\lambda \rightarrow \infty$, and from \eqref{sol Lax 1} we obtain that $C(\lambda)\sim \lambda^k J$. If we choose $C(\lambda) = \lambda^k J$, then, following Lemma \ref{lem split}, we can write the Lax equation in the form $U_{t_{k}} =  \left[ J, a_{-1} \right]$, with 
$$a_{-1} = \text{Res}(a(\lambda)) = \sum_{i+j=k+1} W_i J W_{j}^{T}\,,$$
where $W_i$ are calculated according to Proposition \ref{prop M=WH}. Equation \eqref{D-S n} is invariant under the Lie group
$$
(x, U, W_n, h_n ) \mapsto (\widetilde{x},\widetilde{U}, \widetilde{W}_n, \widetilde{h}_n ) = (e^{\epsilon}x, e^{-\epsilon} U, e^{-n \epsilon} W_n, e^{-(n+1) \epsilon}h_n)\,.
$$
Since $W_n$ is homogeneous of degree $n$, it follows that $a_{-1}$ is homogeneous of degree $k+1$. Therefore the term $U_{t_k}$ is of the same degree. This implies that the expression  $U_{t_{k}} -  \left[ J, a_{-1} \right] $ is homogeneous of degree $k+1$, hence each Lax equation \eqref{Lax eq} admits the symmetry \eqref{scale-k}.
\end{proof}

\begin{example}
For $k=1$ we find 
$$
\mathcal{V}_1(\lambda) = \left( \lambda W(\lambda) J W(\lambda)^T  \right)_{+} = \lambda J + \left( J W_1^T + W_1 J \right) \,.$$
Solving equations \eqref{D-S n} recursively we find $W_1 = \mbox{ad}_J (U)$, hence we have that $\mathcal{V}_1(\lambda) =  \lambda J + U = \mathcal{U}(\lambda) $. Therefore, the Lax equation \eqref{Lax eq} provides the first equation of the hierarchy, ${\bf u}_{t_1} = {\bf u}_x$. In a similar manner, the case $k=3$ provides the next member of the hierarchy. In particular, we obtain the following matrix 
\be \label{V3}
\mathcal{V}_3 = \begin{pmatrix}
0& \lambda^3 - \lambda \frac{ \| {\bf u}\| ^2}{2}   & -\lambda {\bf u}_{x}^T \\
- \lambda^3 + \lambda \frac{ \| {\bf u}\| ^2}{2} & 0 & \lambda^2 {\bf u}^T - \frac{ \| {\bf u}\| ^2}{2} {\bf u}^T - {\bf u}^T_{xx}   \\
 \lambda {\bf u}_{x} &  -\lambda^2 {\bf u}^T + \frac{ \| {\bf u}\| ^2}{2} {\bf u}^T + {\bf u}^T_{xx}   & {\bf u} {\bf u}_{x}^{T} - {\bf u}_{x} {\bf u}^{T}  \\
\end{pmatrix} ,
\ee
such that the compatibility of $\mathcal{A}_3(\lambda) = D_{t_3} - \mathcal{V}_3(\lambda)$  and $\mathcal{L}(\lambda)$ is equivalent to the vmKdV equation \eqref{vmkdv}.
\end{example}

The flows defined by the Lax equations \eqref{Lax eq} commute, namely $[ D_{t_n}, D_{t_m}] \mathbf{u} = {\bf 0}$, for  $n, m = 1, 3, \ldots$.  We prove this below following the ideas in \cite{DS}, also taking into account the reduction group generated by \eqref{t-aut}-\eqref{r-aut}.

\begin{lemma} \label{lem As comm}
The Lax operators $\mathcal{A}_{k}$, $k= 1,3, \ldots$, commute.
\end{lemma}

\begin{proof}
The Lax equation  $[ \mathcal{A}_k, \mathcal{L} ]=0$, in view of  relation \eqref{intertw}, takes the form 
$$[ \mathcal{M}^{-1} \mathcal{A}_k \mathcal{M}, \mathcal{L}_0 ]=0\,, \quad \text{with} \quad  \mathcal{M}^{-1} \mathcal{A}_k \mathcal{M} = D_{t_k} -\mathcal{M}^{-1} \mathcal{V}_k \mathcal{M} + \mathcal{M}^{-1} \mathcal{M}_{t_k} \,.
$$
Since $\mathcal{V}_k(\lambda) = (\lambda^k \mathcal{M}(\lambda) J \mathcal{M}(\lambda)^{-1})_{+}$, it follows that $\mathcal{M}^{-1} \mathcal{A}_k \mathcal{M}= D_{t_k} - \lambda^k J - \hat{\mathcal{V}}_{k}(\lambda)$, with $\hat{\mathcal{V}}_k(\lambda) = - \mathcal{M}^{-1} \mathcal{M}_{t_k} - \mathcal{M}^{-1} (\lambda^k \mathcal{M} J \mathcal{M}^{-1})_{-} \mathcal{M}  \in \mathfrak{g}((\lambda^{-1}))^{ \langle \mathfrak{q} \rangle }$ and of the form 
$$
\hat{\mathcal{V}}_k(\lambda) = \frac{1}{\lambda}\hat{V}_1 + \frac{1}{\lambda^2}\hat{V}_2 + \cdots \,.
$$
The commutativity relation $\left[ \mathcal{M}^{-1} \mathcal{A}_k \mathcal{M} , \mathcal{L}_0 \right] =0$ is equivalent to $D_x \hat{\mathcal{V}}_k = \lambda [J, \hat{\mathcal{V}}_k]$. Comparing powers of $\lambda$ we obtain 
$$
\mbox{ad}_J\hat{V}_1 = 0\,, \quad D_x \hat{V}_i = \mbox{ad}_J \hat{V}_{i+1} \,, \quad i = 1,2,\ldots \,.
$$
Recursively we find that $D_x \hat{V}_i \in \mathfrak{K}^{(i\bmod{2})}$ while $\mbox{ad}_J \hat{V}_{i+1} \in \mathfrak{I}^{(i\bmod{2})}$, therefore $D_x \hat{\mathcal{V}}_k (\lambda) = 0$ and $\hat{\mathcal{V}}_k(\lambda) \in \mathfrak{K}$. We have that $[ \mathcal{A}_n, \mathcal{A}_m ]= [ \mathcal{A}_n, \mathcal{A}_m ]_{+}$, as the Lax operators $\mathcal{A}_n , \, \mathcal{A}_m$ are polynomial in $\lambda$, hence
$$
[ \mathcal{A}_n, \mathcal{A}_m ]_{+} = \left( \mathcal{M} ( D_{t_m} \hat{\mathcal{V}}_n - D_{t_n} \hat{\mathcal{V}}_m + [\hat{\mathcal{V}}_n, \hat{\mathcal{V}}_m ] ) \mathcal{M}^{-1}\right)_{+} = 0 \,.
$$
\end{proof}
The commutativity of the flows $[D_{t_n}, D_{t_m}] {\bf u} = {\bf 0}$ follows from the Lax equations \eqref{Lax eq} and the Jacobi identity. Indeed, we have $[D_{t_n}, D_{t_m}] \mathcal{L} = [ [\mathcal{A}_n, \mathcal{A}_m ], \mathcal{L} ] = 0$ from the above lemma.

\subsection{The recursion operator}

In this section we construct the recursion operator for the vmKdV hierarchy \eqref{Lax eq}. This recursion operator was first derived in \cite{JP vmkdv} and later in \cite{Anco vmkdv} using the bi-Hamiltonian formalism. Here we present an alternative construction, following \cite{Gurses rec}, using the Lax matrices of the hierarchy \eqref{Vk} and its reduction group \eqref{t-aut}-\eqref{r-aut}.

\begin{prop}
The recursion operator $\mathcal{R}$ for the vmKdV hierarchy is given by 
\be \label{recursion op}
\mathcal{R} {\bf f} = - D_x^2 {\bf f} - \| {\bf u} \|^ 2 {\bf f} - {\bf u}_x D_x^{-1} \left( {\bf u}^T {\bf f} \right) -  D_x^{-1} ({\bf u}_x \wedge {\bf f}) {\bf u} \,,
\ee
where ${\bf a} \wedge {\bf b} = {\bf a}{\bf b}^T - {\bf b}{\bf a}^T$.
\end{prop}

\begin{proof}
We split \eqref{Vk} in polynomial and purely negative powers in $\lambda$  as follows
$$\mathcal{V}_{2n +1}(\lambda) = \left( \lambda^2 \lambda^{2n-1} WJW^T \right)_{+} =  \left(  \lambda^2 \mathcal{V}_{2n-1}(\lambda) \right)_{+}   + \left(  \lambda^2 \left( \lambda^{2n-1} W J W^{T} \right)_{-}  \right)_{+}.
$$ 
Hence we can write the following recursive expression for the Lax matrices $\mathcal{V}_{k}(\lambda)$, for $k=1,3, \ldots$,
\be \label{V recur} 
\mathcal{V}_{2 n+1}(\lambda) = \lambda^2 \mathcal{V}_{2 n-1}(\lambda) + \lambda A_{2n-1} + B_{2n-1}\,.
\ee
Since $\mathcal{V}_{2n+1}(\lambda) \,, \mathcal{V}_{2n-1}(\lambda) \in \mathfrak{g}[\lambda]^{\langle \mathfrak{q} \rangle}$  and $\lambda^2$ is invariant under $\lambda \mapsto -\lambda$, it follows that $ A_{2n-1} \in  \mathfrak{g}^{(1)}$ and $B_{2n-1} \in \mathfrak{g}^{(0)}$.  Then from the Lax equations \eqref{Lax eq} we have
$$
\mathcal{L}_{t_{2n+1}} = \lambda^2 \mathcal{L}_{t_{2n-1}} + \lambda [A_{2n-1}  , \mathcal{L}] + [ B_{2n-1}, \mathcal{L} ] \,,
$$ 
from which, comparing powers of $\lambda$, we obtain
\begin{eqnarray}
U_{t_{2n+1}} &=& D_x B_{2n-1} + [B_{2n-1},U] \,, \nonumber \\
D_x A_{2n-1} &=& [U, A_{2n-1}] + [J, B_{2n-1}] \,, \label{A B rec}  \\
U_{t_{2n-1}} &=& [J, A_{2n-1}]  \,. \nonumber
\end{eqnarray} 
The above equations provide the relation between two members of the vmKdV hierarchy, while fixing $A_{2n-1}, B_{2n-1}$ in terms of ${\bf u}$ and ${\bf u}_{t_{2n-1}}$. In particular, we have that
$$
{\bf u}_{t_{2n+1}} = - D_x^2\, {\bf u}_{t_{2n-1}} - \| {\bf u} \|^2 {\bf u}_{t_{2n-1}} - {\bf u}_x D_x^{-1} \left( {\bf u}^T {\bf u}_{t_{2n-1}}   \right)  -  D_x^{-1} ({\bf u}_x \wedge {\bf u}_{t_{2n-1}}) {\bf u} \,.
$$
Hence, we can write ${\bf u}_{t_{2n+1}} = \mathcal{R}  {\bf u}_{t_{2n-1}}$, with $\mathcal{R}$ given in \eqref{recursion op}.
\end{proof}

\begin{example}
Acting with the recursion operator \eqref{recursion op} on the first equation of the hierarchy, ${\bf u}_{t_1} = {\bf u}_x$, we obtain the vmKdV equation
$$
{\bf u}_{t_3} = \mathcal{R} {\bf u}_{t_1} = \mathcal{R} {\bf u}_{x} = - {\bf u}_{xxx} - \fract{3}{2} \| {\bf u} \|^2 {\bf u}_x \,.
$$
Accordingly, we find the next equation of the hierarchy to be
\begin{equation}
{\bf u}_{t_5} =  {\bf u}_{xxxxx}  + \fract{5}{2} \left(  {\bf u}^T{\bf u} \right) {\bf u}_{xxx} + \fract{5}{2} \left(  {\bf u}_x^T{\bf u}_x \right) {\bf u}_{x} + 5 \left(  {\bf u}^T{\bf u}_x \right) {\bf u}_{xx} + 5 \left(  {\bf u}^T{\bf u}_{xx} \right) {\bf u}_{x}  + \fract{15}{8} \left(  {\bf u}^T{\bf u} \right)^2 {\bf u}_{x} \,. \nonumber
\end{equation}
\end{example}

From equation \eqref{V recur} we obtain a recursion relation for the Lax matrices. Using projections $P_{\mathfrak{K}}$ and $P_{\mathfrak{I}}$, introduced in Proposition \ref{prop M sol},  in equations \eqref{A B rec} we find
\begin{eqnarray}
A_{2n-1} &=& - \mbox{ad}_J (U_{t_{2n-1}}) - D_x^{-1}({\bf u}^T {\bf u}_{t_{2n-1}}) J \,, \nonumber \\
B_{2n-1} &=& -D_x U_{t_{2n-1}} - D_x^{-1}({\bf u}^T {\bf u}_{t_{2n-1}}) U + D_x^{-1} \left[ D_x U_{t_{2n-1}}, U \right]  \,.\nonumber
\end{eqnarray}
Hence, relation \eqref{V recur} can be written in matrix form as follows for $n=1,2,\ldots$
\begin{multline}
\mathcal{V}_{2n+1}(\lambda)  = \lambda^2 \mathcal{V}_{2n-1}(\lambda) + \lambda 
\begin{pmatrix}
0 & - D_x^{-1} ({\bf u}^T {\bf u}_{t_{2n-1}}) & - {\bf u}^T_{t_{2n-1}} \\
D_x^{-1} ({\bf u}^T {\bf u}_{t_{2n-1}}) & 0 & {\bf 0}^T \\
{\bf u}_{t_{2n-1}} & {\bf 0} & \bb{0} 
\end{pmatrix} \\
+ \begin{pmatrix}
0 & 0 & {\bf 0}^T \\
0 & 0 & -D_x ({\bf u}_{t_{2n-1}}^T) - D_x^{-1}({\bf u}^T_{t_{2n-1}} {\bf u} ) {\bf u}^T \\
{\bf 0} & D_x ({\bf u}_{t_{2n-1}}) + D_x^{-1}({\bf u}^T{\bf u}_{t_{2n-1}}  ) {\bf u}  & D_x^{-1} ( {\bf u} \wedge D_x({\bf u}_{t_{2n-1}})) 
\end{pmatrix}, \nonumber
\end{multline}
with $\mathcal{V}_1(\lambda)=\lambda J+U$ and $\textbf{u}_{t_1}=\textbf{u}_{x}$.

\subsection{Conservation laws}

The vmKdV hierarchy admits an infinite number of conservation laws. In the present section we construct a generating function for conservation laws using the Lax representation of the hierarchy as presented in the previous section. The conservation laws for the scalar mKdV equation first appeared in \cite{mKdV claws}, while the first few conserved densities and corresponding fluxes related to the vmKdV equation first appeared in \cite{miky titi}. We also point the reader to \cite{Anco sol}, where the complex mKdV and Sasa-Satsuma equations are treated.

\begin{prop} \label{prop claws}
The vmKdV hierarchy admits an infinite number of scalar conservation laws, with the $(1,2)$ element of matrix $h(\lambda)$ in \eqref{h} being a generating function of the corresponding densities.
\end{prop}

\begin{proof}
Given the factorisation \eqref{factor DT} of the Darboux matrix $\mathcal{M}(\lambda)$, the Lax equations $[\mathcal{A}_k, \mathcal{L}]=0$ can be written as
$[W^{-1} \mathcal{A}_k W, H \mathcal{L}_0 H^{-1}] = 0$, where $H \mathcal{L}_0 H^{-1} = D_x - \lambda J - h(\lambda)$ and $W^{-1} \mathcal{A}_k W =  D_{t_k} - \lambda^k J - \widetilde{\mathcal{V}}_{k}(\lambda)$, with $\widetilde{\mathcal{V}}_k(\lambda)$  $\in \mathfrak{g}((\lambda^{-1}))^{ \langle \mathfrak{q} \rangle }$ of the form
$$
\widetilde{\mathcal{V}}_k(\lambda)= \frac{\widetilde{V}_{k,1}}{\lambda} + \frac{\widetilde{V}_{k, 2}}{\lambda^2} + \cdots \,.
$$
We have that $\widetilde{\mathcal{V}}_k = H\hat{\mathcal{V}}_k H^{-1} + H_{t_k} H^{-1}$, with $\hat{\mathcal{V}}_k(\lambda) \in \mathfrak{K}$ as in Lemma \ref{lem As comm}  and $H$ as in Proposition \ref{prop M=WH}. Therefore,  $\widetilde{\mathcal{V}}_k(\lambda) \in \mathfrak{K}$, and thus $[W^{-1} \mathcal{A}_k W, H \mathcal{L}_0 H^{-1}] = 0$ takes the form
\be \label{claws rel}
D_{t_k} h(\lambda) - D_{x} \widetilde{\mathcal{V}}_k(\lambda) + [h(\lambda), \widetilde{\mathcal{V}}_k(\lambda)] = 0 \,.
\ee
Given that $h(\lambda)\,, \widetilde{\mathcal{V}}_k(\lambda) \in \mathfrak{K}$ and in view of relations \eqref{Ker dec}, it follows from \eqref{claws rel} that 
$$
D_{t_k} \big( h(\lambda)_{12} \big)  = D_x \left( \widetilde{\mathcal{V}}_k(\lambda)_{12} \right)  \quad \mbox{for} \quad k= 1,3, \ldots \,,
$$
hence the element $h(\lambda)_{12}$ is a generating function of conserved scalar densities for the vmKdV hierarchy, with the corresponding fluxes given by  $\widetilde{\mathcal{V}}_k(\lambda)_{12}$.
\end{proof}

The first three conserved densities, up to equivalence modulo $\text{Im}(D_x)$, are given by 
$$
\begin{array}{llc}
f_1  = \fract{1}{2} \| {\bf u} \|^2\,, \quad f_2 = - \fract{1}{8} \| {\bf u} \|^4 + \fract{1}{2} \| \mathbf{u}_x\|^2 \,, \\ \\
 f_3   = \fract{1}{2} \| {\bf u} \|^6 + 4 \|{\bf u}_{xx}\|^2 + 6  \| {\bf u} \|^2 {\bf u}^T {\bf u}_{xx}  + 8 ({\bf u}^T {\bf u}_x)^2 \,.
\end{array}
$$

\begin{remark} 
While the conserved densities obtained from the generating function $h(\lambda)_{12}$ are $O_N$-invariant, there exist additional conserved densities which are $O_N$-invariant modulo terms in $\text{Im}(D_x)$. Indeed, there exists a matrix conservation law for the vmKdV hierarchy given by
$$
D_{t_k} ( h_2 )  = D_x ( \widetilde{V}_{k, 2} )  \,,
$$
with $h_2\,, \widetilde{V}_{k, 2} \in \mathfrak{K}^{(0)}$. The corresponding densities are $u_i D_x(u_j)$. Additionally, we observe that contrary to the case of the scalar mKdV equation, the vmKdV hierarchy does not admit ${\bf u}$ as a conserved density. This implies that the conservation of mass no longer holds in the vectorial case. See also \cite{Anco sol} for a discussion.
\end{remark}

\subsection{Special case: $ \mathfrak{g} = \mathfrak{so}_{4}(\mathbb{R})$} \label{sec: so4}

The current section is dedicated to the case where $N=2$, so ${\bf u} = ( u_1, u_2)^T$. In this case the vmKdV equation \eqref{vmkdv} is equivalent to the system of two equations
$$
u_{1 t} + u_{1 xxx} + \fract{3}{2}(u_1^2 + u_2^2) u_{1 x} = 0 \,, \quad  u_{2 t} + u_{2 xxx} + \fract{3}{2}(u_1^2 + u_2^2) u_{2 x} = 0\,.
$$
After setting $v = u_1 + i u_2$, we find that the field $v=v(x,t)$ satisfies the complex mKdV equation, also known as Hirota equation \cite{Hirota eq},
\be \label{cmKdV}
v_{t} + v_{xxx} + \fract{3}{2} |v|^2 v_{x} = 0 \,, \quad 
\mbox{where} \quad  |v|^2 = v^* v = u_1^2 + u_2^2 \,.
\ee
In Proposition \ref{prop Vk} we showed that the Lax equations $D_{t_k} \mathcal{L} = [  \mathcal{V}_k, \mathcal{L}  ]$ define the vmKdV hierarchy for $\mathcal{V}_{k}(\lambda) = \left( \lambda^k W J W^T \right)_{+}$, with $k$ odd. For generic $N$, if $k$ is even then the coefficients of $\mathcal{V}_k(\lambda)$ contain non-local elements, hence leading in general to non-local Lax equations. More specifically, the first non-local elements appear in the coefficient of $\lambda^{k-2}$ and are of the form $D_x^{-1} ({\bf u} {\bf u}_x^T {\bb A} + {\bb A} {\bf u}_x {\bf u}^T )$, with ${\bb A} \in \mathfrak{so}_N$. However, when $N=2$ we obtain local equations for even $k$, since in this case the term ${\bf u} {\bf u}_x^T {\bb A} + {\bb A} {\bf u}_x {\bf u}^T$ is in the image of $D_x$ and  proportional to $\bb A$, as the basis of $\mathfrak{so}_2$ is one dimensional.

\begin{prop}
When $\mathfrak{g} = \mathfrak{so}_{4}(\mathbb{R})$ the hierarchy generated by the Lax operator $\mathcal{L}(\lambda)$ in \eqref{L} contains the nonlinear Schr{\"o}dinger equation. 
\end{prop}

\begin{proof}
When $k=2$  the Lax matrix $\mathcal{V}_k(\lambda)$ is of the form $\mathcal{V}_2(\lambda) = \lambda^2a_2 + \lambda a_1 + a_0$, with  $a_1 \in \mathfrak{g}^{(1)}$ and $a_0, \, a_2 \in \mathfrak{g}^{(0)}$. From equations \eqref{a i} we have that $a_2 \in \mathfrak{K}^{(0)}$, hence is of the form
$$ 
\begin{pmatrix}
0 & 0 & {\bf 0}^T \\
0 & 0 & {\bf 0}^T \\
{\bf 0} & {\bf 0} & \bb{A}  \\
\end{pmatrix}, \quad \mbox{with} \quad {\bb A} = \begin{pmatrix} 0 & 1\\ -1 & 0\\ \end{pmatrix} \: \in \mathfrak{so}_2\,,
$$
while for $a_1$ and $a_0$ we find
$$
a_1 = \begin{pmatrix}
0 & 0 & {\bf u}^T {\bb A} \\
0 & 0 & {\bf 0}^T \\
 {\bb A \bf u} &  \bf 0 & \bb{0}  \\
\end{pmatrix} \,,
~
a_0 = \begin{pmatrix}
0 & 0 &  {\bf 0}^T\\
0 & 0 & {\bf u}_x^T {\bb A} \\
\bf 0 & {\bb A} {\bf u}_x & -\fract{1}{2} \| {\bf u} \|^2\bb{A}  \\
\end{pmatrix}\,.
$$ 
Then, the Lax equation \eqref{Lax eq} is equivalent to the system of equations
$$
u_{1 t} + u_{2 xx} + \fract{1}{2}(u_1^2 + u_2^2) u_2 = 0 \,, \quad  u_{2 t} - u_{1 xx} - \fract{1}{2}(u_1^2 + u_2^2) u_1 = 0\,.
$$
Setting $v = u_1 - i u_2$ we find that the complex field $v=v(x,t)$ satisfies the NLS equation
\be \label{nls}
i v_t - v_{xx} - \fract{1}{2} |v|^2 v = 0 \,, \quad 
\mbox{where} \quad  |v|^2 = v^* v \,.
\ee
\end{proof}

Further, when $N=2$ there also exist additional conservation laws for the vmKdV hierarchy. Indeed, since $\widetilde{\mathcal{V}}_k(\lambda), h(\lambda) \in \mathfrak{K}$ and $\mathfrak{K}$ is abelian when $N=2$, \eqref{claws rel} takes the form of a conservation law
$$
D_{t_k} h(\lambda) = D_x \widetilde{\mathcal{V}}_{k}(\lambda) \quad \mbox{for} \quad k= 1, 2, 3, \ldots \,.
$$
For example, we obtain the conserved density 
$$
f_4 = \fract{1}{2} \| {\bf u} \|^2 u_1u_{2 x} + u_1 \left( u_{2 xxx} +  u_2^2u_{2 x}  \right) \,.
$$
In addition, in the $N=2$ case there exists a relation between the recursion operator $\mathcal{R}_{NLS}$ of the NLS equation and that of the vmKdV equation $\mathcal{R}$ given in \eqref{recursion op}. This relation was also remarked in \cite{JP vmkdv}.

\begin{prop}
The recursion operators of the two-component vmKdV \eqref{recursion op} and the NLS equations \eqref{nls} are related according to $\mathcal{R}_{NLS}^2= \mathcal{R}$.
\end{prop}
\begin{proof}
The recursion operator of the (self-focusing) NLS equation is 
$$
\mathcal{R}_{NLS} f = - i D_x f - \fract{i}{2} \, u \,D_x^{-1} \left( f u^* + u f^* \right).
$$
Setting $u=u_1+i u_2$,  and considering the action of $\mathcal{R}_{NLS}$ in real plane vectors, we can represent $\mathcal{R}_{NLS}$ in matrix form as
$$
\mathcal{R}_{NLS} =
\begin{pmatrix}
u_2 D_x^{-1} u_1 & D_x + u_2 D_x^{-1} u_2  \\
-D_x - u_1 D_x^{-1} u_1 & - u_1 D_x^{-1} u_2 \\
\end{pmatrix}.
$$
The recursion operator of the vmKdV hierarchy  \eqref{recursion op}  in the case where $N=2$ and ${\bf u} = ( u_1, u_2)^T$ can be expressed in matrix form as
$$
\mathcal{R}= - \begin{pmatrix}
D_x^2 +   \| {\bf u} \|^2\ + u_{1 x} D_x^{-1} u_1 - u_2 D_x^{-1}u_{2 x} &   u_{1 x} D_x^{-1} u_2 + u_2 D_x^{-1}u_{1 x}   \\
 u_{2 x} D_x^{-1} u_1 + u_1 D_x^{-1}u_{2 x} & D_x^2 +  \| {\bf u} \|^2\ + u_{2 x} D_x^{-1} u_2 - u_1 D_x^{-1}u_{1 x} \\
\end{pmatrix}
$$
Finally, it can be verified that $\mathcal{R}_{NLS} ^2 = \mathcal{R}$.
\end{proof}

\section{Soliton solutions for the vector mKdV hierarchy} \label{sec: dressing}

The vmKdV hierarchy admits both soliton and breather solutions. In this section we present the Darboux transformations and the corresponding dressing formulas that produce such solutions. A Darboux-dressing matrix $M(\lambda)$ defines the transformations
\be \label{Lax - D}
\widetilde{\mathcal{L}}(\lambda) = M(\lambda) \mathcal{L}(\lambda)M(\lambda) ^{-1} \,, \quad \widetilde{\mathcal{A}}_k(\lambda) = M(\lambda) \mathcal{A}_k(\lambda)  M(\lambda)^{-1} \,,
\ee
where in \eqref{Lax - D} $\mathcal{L}(\lambda)$, $\mathcal{A}_k(\lambda)$ are as in \eqref{L}, \eqref{A}, while $\widetilde{\mathcal{L}}(\lambda) = \mathcal{L}(\widetilde{\mathbf{u}}, \lambda)$, $\widetilde{\mathcal{A}}_k(\lambda) = \mathcal{A}_k(\widetilde{\mathbf{u}}, \lambda)$,  and $\widetilde{\mathbf{u}}$ is the `dressed' solution of the vmKdV hierarchy. Additionally, the Darboux matrices $M(\lambda)$ of the vmKdV hierarchy satisfy
\be  M(\lambda) M(\lambda)^{T} = \mathds{1}\,, \quad  M(\lambda)  =  Q M(-\lambda) Q^{-1} \,,  \quad M(\lambda) = M(\lambda^*)^{*}  \label{M prop} \,,
\ee
which follow from the reduction group \eqref{t-aut}-\eqref{r-aut}  of the Lax operators. Such types of Darboux transformations were also used in \cite{vSG dress, vSG DT} for the study of soliton and breather solutions, their interactions, and related integrable structures associated with the vector sine-Gordon equation. 

The Darboux matrices $M(\lambda)$ have rational dependence on the spectral parameter $\lambda$, and have simple poles. The soliton solutions correspond to Darboux matrices with purely imaginary poles, while the Darboux matrices for the breather solutions have poles at generic points in the complex plane. Moreover, each of these Darboux matrices is parametrised by an element of the Grassmannian $Gr(s,{\bb C}^{N+2})$, with $s=1, 2, \ldots, N+1$ in the case of a breather solution, and $s=1$ for a soliton solution. 

Given that $M(\lambda)$ satisfies \eqref{M prop}, it follows that it will have poles at group orbits of the group $\mathbb{Z}_2 \times \mathbb{Z}_2$ generated by $\mathfrak{q}$ and $\mathfrak{r}$. We focus on Darboux matrices that have poles in only one orbit, and we call them elementary Darboux matrices. A Darboux matrix, with poles at a generic orbit $\lbrace \mu, -\mu, \mu^*, -\mu^* \rbrace$, of the form 
\be \label{M br}
M(\lambda) = \mathbb{1} + \frac{M_0}{\lambda - \mu} - \frac{QM_0Q}{\lambda + \mu} + \frac{M^*_0}{\lambda - \mu^*} - \frac{QM_0^*Q}{\lambda + \mu^*},
\ee
leads to a breather solution. A soliton solution corresponds to a degenerate orbit $\lbrace i \mu, - i \mu \rbrace$ formed by a purely imaginary pole with a Darboux matrix given by
\be \label{M soli}
M(\lambda) = \mathbb{1} + \frac{M_0}{\lambda - i \mu} -  \frac{QM_0Q}{\lambda +i  \mu} \,, \quad \mbox{with} \quad M_0 = - QM_0^*Q \,.
\ee
One can verify that matrices of the form \eqref{M br} and \eqref{M soli} are invariant under $\mathfrak{q}$ and $\mathfrak{r}$. The orthogonality condition in \eqref{M prop} will be imposed separately for each of the Darboux matrices given above.  We define the rank of a soliton or breather solution to be the rank of matrix $M_0$. In what follows we construct the Darboux-dressing matrix for the one soliton and breather solutions for the vmKdV hierarchy, using the dressing formula
\be \label{dressing res}
\widetilde{U} = U - \mbox{ad}_J \left(\underset {\lambda = 0}{\Res} M(\lambda) \right),
\ee
which follows from \eqref{Lax - D}.

\subsection{Soliton solution}

In this section, we use the Darboux matrix \eqref{M soli} in order to derive the one soliton solution for the vmKdV hierarchy. We find that the one soliton solution is parametrised by a point in the imaginary line (the pole of the Darboux matrix \eqref{M soli}) and a point in the sphere $\mathbb{S}^N$. 

\begin{prop}
A Darboux matrix with poles on a degenerate orbit, i.e. of the form \eqref{M soli}, satisfies the orthogonality condition $M(\lambda) M(\lambda)^T = \mathbb{1}$ if and only if 
\be \label{DT P} 
M_0 = 2 i \mu P \quad \mbox{with} \quad P = \frac{Q \mathbf{q} \mathbf{q}^T}{\mathbf{q}^T Q \mathbf{q}} \quad \mbox{and} \quad \mathbf{q}^T \mathbf{q} = 0 \,,
\ee
where $\mathbf{q} \in \mathbb{C P}^{N+1}$.
\end{prop}

\begin{proof}
The double pole at $\lambda =i \mu$ of the orthogonality condition implies that $M_0 M_0^T = 0$.  It follows that $M_0$ is not of full rank, and together with the reality condition $M_0 = -QM_0^*Q$ we have that $\rank (M_0) = 1$. For a proof see \cite{vSG dress}, and \cite{vSG DT} for a discussion. In this case, $M_0$ can be parametrised by two $N+2$ dimensional complex vectors $\mathbf{p}, \mathbf{q}$ as 
$$M_0 = \mathbf{p} \mathbf{q}^T\,, \quad  \text{with} \quad \mathbf{q}^T \mathbf{q} = 0 \,.$$ 
The residue at $\lambda = i\mu$ of relation $M(\lambda) M(\lambda)^T = \mathbb{1}$ implies the relation
$$
\mathbf{p} = \frac{2 i \mu}{\mathbf{q}^T Q \mathbf{q}} Q \mathbf{q}\,.
$$ 
The matrix $P$  is invariant under the scaling $\textbf{q}\mapsto \alpha\textbf{q}$ with $\alpha$ a complex valued function of the independent variables. Therefore, ${\bf q}$ is an element in the Grassmannian $Gr(1,\mathbb{C}^{N+2})\simeq  \mathbb{C P}^{N+1}$.
\end{proof}

\begin{remark}
From the proposition above it follows that the Darboux matrix in \eqref{M soli} takes the form
\be \label{DT soliton}
M(\lambda) = \mathbb{1} + \frac{2 i \mu}{\lambda - i \mu} P -  \frac{2 i \mu}{\lambda +i  \mu}Q P Q \,, \quad \mbox{with} \quad P^{*} = Q P Q \,.
\ee
We observe that the matrix $P$ is a projector, i.e. satisfies $P^2=P$. 
\end{remark}

From the double pole at $\lambda =i  \mu$ of the Lax-Darboux equations \eqref{Lax - D} we obtain that the vector $\bf{q}$ satisfies 
$$
\mathcal{L}(i \mu) {\bf q} = 0, \quad \mathcal{A}_{k} (i \mu) {\bf q} =0 \,.
$$ 
Thus, we have that
\be\label{q final sol}
{\bf q} = \Psi(i \mu) {\bf C} \,, \quad \rm{with} \quad {\bf C}^T {\bf C} = 0 \,,
\ee
where $\Psi(i \mu)$ is the common fundamental solution of the Lax operators of the vmKdV hierarchy at  $\lambda = i \mu$.

Given the form \eqref{DT soliton} of the Darboux matrix $M(\lambda)$, the dressing formula \eqref{dressing res} can be written as
\be \label{dress soliton}
\widetilde{U} = U - 2 i \mu \, \mbox{ad}_J \left(P - P^*\right).
\ee
Then, taking into account the structure of $U$ in \eqref{J U} and $P$ in \eqref{DT P}, the dressing formula \eqref{dress soliton} for the vmKdV hierarchy takes the form
\be \label{dressing sol}
\widetilde{u}_j = u_j - 4 i \mu \,  \frac{q_1 \,q_{j+2}}{-q_1^2 + \sum_{k=2}^{N+2} q_k^2} \,, \quad j = 1,2, \ldots, N \,,
\ee
with $q_j$, $\widetilde{u}_j$ the components of vector ${\bf q}$ and transformed vmKdV field $\widetilde{{\bf u}}$, respectively.

\begin{example}
We start with the trivial solution ${\bf u}_0 ={\bf 0}$. In this case the Lax operators of the vmKdV hierarchy at $\lambda = i \mu$ take the form 
$$
\mathcal{L}(i \mu) = D_x - i \mu J \,, \quad \mathcal{A}_k(i \mu) = D_{t_k} - (i \mu)^k J \,, \quad \mbox{with} \quad k=1, 3, \ldots \,.
$$
Hence, the fundamental solution for the above systems of operators is
$$
\Psi(i \mu)  =  
\begin{pmatrix}
\cosh \xi & i \sinh \xi & {\bf 0}^T \\
-i \sinh \xi & \cosh \xi & {\bf 0}^T \\
\bf 0 &  {\bf 0} & \bb{1} \\ 
\end{pmatrix}\,, \quad \textrm{with} \quad \xi(\mu, t_k)= \sum_{n=0}^{\infty} (-1)^n \mu^{2 n+1} t_{2 n +1}\,.
$$
Then, from expression \eqref{q final sol} and the reality condition $P^* = QPQ$ in \eqref{DT soliton} we obtain that  ${\bf C} = (i, c_0, {\bf c}^T)^T$, with $c_0 \in \mathbb{R}$, and ${\bf c} \in \mathbb{R}^N$ a constant vector such that $c_0^2+c_1^2+ \cdots +c_N^2 = 1$. Thus, it follows that the vector ${\bf q}$ can be expressed as
\be \label{q Psi}
{\bf q}= (i \cosh \xi+ i c_0 \sinh\xi, c_0 \cosh\xi + \sinh\xi, {\bf c}^T )^T \,,
\ee
and then expression \eqref{dressing sol} leads to the one-soliton solution for the vmKdV hierarchy
$$
{\bf u} = \frac{2 \mu \,  {\bf c}}{\cosh \xi+  c_0 \sinh\xi} \,.
$$
In the case when $N=1$ and $N=2$ the above reduces to the one soliton solution of the mKdV and complex mKdV equations, respectively, see also \cite{Anco sol} and references therein. Higher soliton solutions can be recursively constructed using the dressing formula \eqref{dress soliton}.
\end{example}

\subsection{Breather solution}

Here, we consider a Darboux matrix $M(\lambda)$ of the form \eqref{M br}, and we present the corresponding dressing formula that leads to breather-type solutions for the vmKdV hierarchy. As an example, we derive the simplest one breather solution, we express it as a ratio of determinants, and we find that it is parametrised by two constant unit vectors normal to each other.

\begin{prop}
A Darboux matrix with poles on a generic orbit, i.e. of the form \eqref{M br}, satisfies the orthogonality condition $M(\lambda) M(\lambda)^T = \mathbb{1}$ if and only if 
\be \label{{M0 br}}
M_0 =  q^*B q^T + Q  q C q^T + Q q^* D q^T \,, \quad \mbox{with} \quad q^T q = 0 \,,
\ee
and $q \in Gr(s,{\bb C}^{N+2}) \simeq M_{N+2,s}(\mathbb{C})/ GL_{s}(\mathbb{C})$, with $s=1, 2, \ldots, N+1$. Additionally, $D, B, C \in M_{s,s}(\mathbb{C})$ are of the form
\be \label{bcd}
D =  - \left(  F H^{-1} F^* + G^* H^{* -1} G^* - H^* \right) ^{-1}, ~ B = D G^* H ^{* -1}, ~  C = -D^* F^*  H ^{* -1} \,,
\ee
where $F, G, H$ are given by
\be \label{fgh}
F = \frac{q^T Q q}{2 \mu} \,, \quad G = \frac{ q^{* T} q }{\mu - \mu^*} \,, \quad H = \frac{q^{*T} Q q}{\mu + \mu^*} \,.
\ee
\end{prop} 

\begin{proof}
The double pole at $\lambda = \mu$ of the orthogonality condition implies that $M_0 M_0^T = 0$. Hence, $M_0$ is not of full rank, i.e. $\rank(M_0) = s$ with $s=1,2, \ldots, N+1$. We parametrise $M_0$ in terms of two matrices $p$, $q \in M_{N+2,s}(\mathbb{C})$ as
$$
M_0 = p q^T \,, \quad \mbox{with} \quad q^T q = 0 \,.
$$
From the residue at $\lambda = \mu$ of relation $M(\lambda) M(\lambda)^T = \mathbb{1}$ it follows that 
\be \label{p ito q}
p = q^*B + Q  q C + Q q^* D \,,
\ee
with $B,C,D$ as in \eqref{bcd}. The matrix $M_0$ is invariant under transformation $q \mapsto q A$, where $A \in M_{s,s}(\mathbb{C})$ and $\det A \neq 0$, hence $q \in Gr(s,{\bb C}^{N+2})$.
\end{proof}

We write the matrix $q$ in terms of $N+2$ dimensional vectors ${\bf q}^l = (q_1^l, q_2^l, \ldots, q_{N+2}^l)^T$, for $l = 1, \ldots, s$,  as $q = ({\bf q}^1, {\bf q}^2, \ldots, {\bf q}^s)$. Then, using relation \eqref{p ito q} we express the components of $p$ in terms of the components of matrix $q$ according to
\be \label{pil}
p_i^l = \sum_{k=1}^s \left( q_i^{* k} B_{k l} - q_i^k C_{k l} - q_i^{* k} D_{k l} \right) \,, \quad i = 1, 2, \ldots, N+2 \,.
\ee

Similar to the soliton case considered in the previous section, from the double pole at $\lambda= \mu$ of the Lax-Darboux equations \eqref{Lax - D} it follows that 
$$
\mathcal{L}(\mu) q = 0, \quad \mathcal{A}_{k} (\mu) q =0\,, 
$$ 
thus, we can express the matrix $q$ in terms of the fundamental solutions of the linear problems $\mathcal{L}(\mu) \Psi =0\,, ~ \mathcal{A}_k(\mu) \Psi = 0$ as
\be\label{q final}
q = \Psi(\mu) R \,, \quad \rm{with} \quad R^T R = 0 \,,
\ee
and $R \in M_{N+2,s}(\mathbb{C})$ a constant matrix.

Expression \eqref{dressing res} leads to the following dressing formula
\be \label{dressing br}
\widetilde{U} = U + [K - QKQ,J] \quad \mbox{with} \quad K = p q^T + p^* q^{*T} = 2 \Re \sum_{l=1}^s {\bf p}^l {\bf q}^{l T} \,,
\ee
which can be written as
\be \label{dressing br u}
\widetilde{u}_j = u_j + 4 \Re \left( \sum_{l=1}^s p^l_1 q^l_{j+2} \right) \,, \quad j = 1,2, \ldots, N \,.
\ee
Using expression \eqref{pil} for $i = 1$, we obtain the dressing transformation for the vmKdV hierarchy, leading to breather-type solutions,
\be \label{dressing br det}
\widetilde{u}_j = u_j - 4 \Re \sum_{k,l = 1}^s 
\begin{vmatrix}
q_1^k & 0 & 0 \\
0 & q^l_{j+2} & B^*_{kl} - D^*_{kl} \\
0 & q^{* l}_{j+2} & C_{kl}
\end{vmatrix}
\,, \quad j = 1,2, \ldots, N \,,
\ee
with $B,C,D$ given in \eqref{bcd}.

\begin{example}
In the case $s=1$, expression \eqref{p ito q} takes the form
\be \label{p br}
\mathbf{p} = \frac{1}{\Delta} \left(  G \mathbf{q}^* + F^* Q \mathbf{q} - H Q \mathbf{q}^*  \right) \,, \quad \Delta = G^2 -H^2 + |F|^2 \,,
\ee
where $G, F, H$ are given in \eqref{fgh} and are now scalar quantities. Then, the dressing transformation \eqref{dressing br det} becomes
$$
\widetilde{u}_j = u_j - 4 \Re \frac{\Delta_j}{\Delta}
\,, \quad j = 1,2, \ldots, N \,,
$$
with $\Delta_j$ and $\Delta$ the following determinants 
$$
\Delta_j = \begin{vmatrix}
q_1 & 0 & 0 \\
0 & q_{j+2} & H-G \\
0 & q^{*}_{j+2} & F^*
\end{vmatrix}, 
\quad 
\Delta = \begin{vmatrix}
F & H-G \\
G+H & F^* 
\end{vmatrix}.
$$
Starting with the trivial solution ${\bf u}_0 = {\bf 0}$, we have that 
$$
\Psi(\mu)  =  
\begin{pmatrix}
\cos \xi & \sin \xi & {\bf 0}^T \\
-\sin \xi & \cos \xi & {\bf 0}^T \\
\bf 0 &  {\bf 0} & \bb{1} \\ 
\end{pmatrix}, \quad \textrm{with} \quad \xi = \sum_{n=0}^{\infty}  \mu^{2 n+1} t_{2 n +1}\,,
$$
hence, from \eqref{q final} we obtain the one breather solution
$$
{\bf u} = - \frac{4}{\Delta} \Re \big( (R_1 \cos \xi + R_2 \sin \xi) (F^* {\bf r} + (G-H) {\bf r}^*)  \big)\,,
$$
where ${\bf R} = (R_1, R_2, {\bf r}^T)^T$ such that ${\bf R}^T {\bf R} = 0$. The latter condition implies that the real and imaginary parts of vector ${\bf R}$ have the same length, and furthermore they are normal to each other. Using the fact that ${\bf R}$ is in $\mathbb{C P}^{N+1}$ we can normalise its real and imaginary parts and assume their length is equal to one. It follows that the one breather solution for the vmKdV hierarchy is parametrised by a complex number (the pole of the Darboux matrix \eqref{M br}) and an element of the unit tangent bundle $T_1 \left(\mathbb{S}^{N+1} \right)$ of the sphere
$$
T_1 \left(\mathbb{S}^{N+1} \right) = \lbrace ( {\bf v}_1, {\bf v}_2) \in \mathbb{R}^{2(N+2)} | \; \left<{\bf v}_1, {\bf v}_2 \right>=0\,, \;\|{\bf v}_1\| =  \|{\bf v}_2\| =1   \rbrace \,.
$$
\end{example}

\section{Conclusions}
In this paper we studied the hierarchy and solutions of a vector mKdV equation which is not a member the vector NLS hierarchy.  An interesting characteristic of the vector mKdV equation that we studied here is that it is an example of an integrable equation admitting a Lax operator that contains a constant non-regular element of the underlying Lie algebra. Furthermore, the Lax operator admits a reduction group isomorphic to $\mathbb{Z}_2 \times \mathbb{Z}_2 \times \mathbb{Z}_2$, as opposed to the case of vector NLS equation. These properties of the Lax operator are important ingredients in the construction of the hierarchy and associated conservation laws, as well as the recursion operator. The reduction group plays also an important role in the construction of solutions to the hierarchy. Specifically, in this work we constructed Darboux transformations for soliton and breather solutions for the whole hierarchy, and presented general formulas for the one soliton and higher-rank one breather solution. An interesting direction of future study would be the construction of higher soliton and breather solutions and their interactions, the expression of the general n-soliton solution in terms of determinants, and, consequently, connections with Hirota's $\tau$ functions and bilinear form for the hierarchy.

\bigskip

\bigskip

\noindent {\bf Acknowledgments} - We would like to thank A. Doikou, A.P. Fordy, A.V. Mikhailov, and B. Vicedo for useful discussions. G.P. greatly acknowledges the support by the Engineering and Physical Sciences Research Council grant EP/P012655/1.

\end{document}